\documentclass[11pt,titlepage]{article}  

\usepackage{color}

\usepackage[margin=1in]{geometry}
 \usepackage[small,compact]{titlesec}
 \usepackage[small,it]{caption}
  \usepackage{times}

\usepackage{algorithmic}
\usepackage[boxed]{algorithm}
\usepackage{amsmath} 
\usepackage{graphicx}
\usepackage{amsthm}
\usepackage{amsfonts}
\usepackage{float}
\floatstyle{boxed} 
\restylefloat{figure}

\def\argmin{\mathop{\arg\,\min}\limits}%

\begin{document}  
\newtheorem{claim}{Claim}
\newtheorem{theorem}{Theorem}
\newtheorem{lemma}{Lemma}
\newtheorem{corollary}{Corollary}

\title{Approximation Algorithms for 
Stochastic Boolean Function Evaluation and Stochastic Submodular Set Cover}
\author{Amol Deshpande \\  University of Maryland \\ amol@cs.umd.edu \and Lisa Hellerstein \\ Polytechnic Institute of NYU \\hstein@poly.edu \and Devorah Kletenik \\Polytechnic Institute of NYU \\ dkletenik@cs.poly.edu}

\maketitle
\section{Introduction}
We present approximation algorithms for two problems: Stochastic Boolean Function Evaluation and Stochastic Submodular Set
Cover.  We also consider a related ranking problem.

Stochastic Boolean Function Evaluation (SBFE) is the problem of determining
the value of a given Boolean function $f$ on an unknown input $x$,
when each bit $x_i$ of $x$ can only be determined by paying a
cost $c_i$.  The assumption is that $x$ is drawn from a
given product distribution, and the goal is to minimize expected
cost.  
SBFE  problems arise in diverse application areas. 
For example, in medical diagnosis, 
the $x_i$ might
correspond to medical tests, 
and $f(x) = 1$ if the patient has a particular disease.
In database query optimization, $f$ could correspond to a 
Boolean query on predicates corresponding to $x_1, \ldots, x_n$,
that has to be evaluated for every tuple in the database to find tuples that satisfy the query~\cite{journals/tods/IbarakiK84,KrishnamurthyBZ86,conf/icde/DeshpandeH08,SrivastavaMWM06}.
In Operations Research, the SBFE problem is known as
``sequential testing'' of Boolean functions.
In learning theory, the SBFE problem
has been studied in the context of
learning with attribute costs.

We focus on developing approximation algorithms
for SBFE problems.
There have been previous papers on exact
algorithms for these problems,
but there
is very little work on approximation algorithms~\cite{unluyurtReview,kaplanMansour-Stoc05}.
Our approach is to reduce the SBFE problems
to Stochastic Submodular Set Cover (SSSC). 
The SSSC problem was introduced by Golovin and Krause, who gave an approximation algorithm
for it called {\em Adaptive Greedy}.
\footnote{Golovin and Krause called the problem Stochastic Submodular Coverage, not Stochastic Submdodular Set Cover, because the cover is not formed using sets. Our choice of name is
for consistency with terminology of Fujito~\cite{fujitoSurvey}.}
Adaptive Greedy is a generalization
of the greedy algorithm for the classical Set Cover problem. We present a new algorithm for
the SSSC problem,
which we call Adaptive Dual Greedy. It is
an extension of the Dual Greedy algorithm for Submodular Set
Cover due to Fujito, which is a generalization of Hochbaum's primal-dual algorithm
for the classical Set Cover 
Problem~\cite{fujitoOriginal,fujitoSurvey}.  We also give a new bound
on the approximation achieved by the Adaptive Greedy algorithm of 
Golovin and Krause.

%


The following is a summary of our results.  We note that
our work also suggests many open questions, including
approximation algorithms for other classes of Boolean functions, 
proving hardness results, and determining adaptivity gaps. 




{\bf \noindent The $Q$-value approach:} 
We first show how to solve the SBFE problem
using the following basic approach, which we call the $Q$-value
approach.
We reduce the SBFE problem to an SSSC problem,
through the construction of an
{\em assignment feasible} utility function, with goal value $Q$.
Then we
apply the
Adaptive Greedy algorithm of Golovin and Krause to the SSSC problem, 
yielding an approximation factor of 
$(\ln Q + 1)$. 

Using this approach, we easily obtain
an $O(\log kd)$-approximation algorithm for CDNF formulas (or decision trees), 
where $k$ is the number of clauses in
the CNF and $d$ is the number of terms in the DNF.
Previously,
Kaplan et al.\ gave an algorithm also achieving
an $O(\log kd)$ approximation, but only
for {\em monotone} CDNF
formulas, unit costs, and  the
uniform distribution~\cite{kaplanMansour-Stoc05}.\footnote{
Although our result 
solves a more general problem than Kaplan et al., 
they give their $O(\log kd)$ approximation factor in terms of 
expected {\em certificate cost},
which lower bounds the
expected cost of the optimal strategy. See Section~\ref{sec:background}.} 


We also use the $Q$-value approach to develop an 
$O(\log D)$-approximation algorithm for
evaluating linear threshold formulas with integer coefficients.
Here $D$ is the sum of the magnitudes of the coefficients. 
This $O(\log D)$ bound is a weak bound that we improve below, but we adapt the algorithm later to obtain other results.

The $Q$-value approach has inherent limitations. 
We prove that it will not give an algorithm 
with a sublinear approximation factor for
evaluating read-once DNF  (even though there is a poly-time
exact algorithm~\cite{kaplanMansour-Stoc05,Greiner02}), or for evaluating
linear threshold formulas with exponentially large coefficients.
In fact, our
weak $O(\log D)$ approximation factor for linear threshold formulas
cannot be improved to be sublinear in $n$ with the $Q$-value approach.
We prove our negative results by
introducing a new combinatorial measure of a Boolean function, 
which we call its $Q$-value.

\smallskip
{\bf \noindent Adaptive Dual Greedy:} 
We present 
Adaptive Dual Greedy (ADG), a new algorithm for the SSSC problem.
We prove that it achieves an approximation factor of
of $\alpha$, where $\alpha$ is a ratio that 
depends on the cover constructed by the algorithm.

\smallskip
{\bf \noindent 3-Approximation for Linear Threshold Formulas:}
We substitute ADG for Adaptive Greedy in our $O(\log D)$ algorithm for evaluating
linear threshold formulas.
We show that in this case, $\alpha$ is
bounded above by 3, and we get a 3-approximation algorithm.

\smallskip
{\bf \noindent New bound on Adaptive Greedy:}
We prove that Adaptive Greedy
achieves an $2(\ln P + 1)$-approximation
in the binary case (and $k(\ln P + 1)$ in the $k$-ary case)
where $P$ is the maximum utility that can
be contributed by a single item.
The proof of this bound uses the same LP that we 
use in our analysis of ADG to lower bound
the approximation factor, combined with Wolsey's approach to bounding
the analogous algorithm for
(non-adaptive) submodular set cover~\cite{wolsey}.
Our bound generalizes Wolsey's
bound for (non-adaptive) submodular set cover~\cite{wolsey},
except for an additional factor of 2.  Wolsey's bound
generalized the $(\ln s + 1)$ bound for 
standard set cover, where
$s$ is the maximum size of one of the input subsets (cf.~\cite{fujitoSurvey}).

\smallskip
{\bf \noindent Simultaneous Evaluation of Linear Threshold Formulas:}
We apply the above techniques
to the problem of simultaneous evaluation of
$m$ linear threshold formulas, giving 
two algorithms
with approximation factors of $O(\log mD_{avg})$ and $D_{max}$
respectively.
Here $D_{avg}$ 
and $D_{max}$ are the average and maximum,
over the $m$ formulas, of the sum of the magnitude of the coefficients. 
These results generalize results of Liu et al.\ for {\em shared filter ordering}~\cite{liuetal}. We also improve one of Liu's results for that problem.

\smallskip
{\bf \noindent Ranking of Linear Functions:}
We give an $O(\log (mD_{max}))$-approximation
algorithm for ranking a set of 
$m$ {\em linear} functions $a_1 x_1 + \ldots + a_n x_n$
(not linear {\em threshold} functions), defined over $\{0,1\}^n$, 
by their output values, in our stochastic setting.
This problem arises in Web search and database query processing.
For example, we might
need to rank a set of documents or tuples by their ``scores'', where
the linear functions compute the scores
over a set of unknown properties such as user preferences or data
source reputations.



\section{Stochastic Boolean Function Evaluation and Related Work}
\label{sec:background}

Formally, the input to the SBFE problem
is a representation of a Boolean function $f(x_1, \ldots, x_n)$ from a fixed class of representations $C$,
a probability vector $p = (p_1, \ldots, p_n)$, where $0 < p_i < 1$, 
and a real-valued cost vector $(c_1, \ldots, c_n)$, where $c_i \geq 0$.
An algorithm for this problem must
compute and output the value of $f$ on an $x \in \{0,1\}^n$, drawn
randomly from product distribution $D_p$, such that $p_i = Prob[x_i = 1]$. 
However, it is not given access to $x$.  Instead, it can
discover the value of any $x_i$ by ``testing'' it, at a cost of $c_i$.  
The algorithm must perform the tests sequentially, 
each time choosing the next test to perform. 
The algorithm can be adaptive, so
the choice of the next test can depend on the
outcomes of the previous tests.
The expected cost of the algorithm is the cost it incurs on a random $x$ from $D_p$.
(Since each $p_i$ is strictly between 0 and 1, the algorithm must
continue doing tests until it has obtained a 0-certificate or 1-certificate
for the function.)
The algorithm is optimal if it
has minimum expected cost with respect to $D_p$.
The running time of the algorithm is the (worst-case) time it takes to
determine the next variable to be tested, or to compute the value of $f(x)$ after
the last test. 
The algorithm corresponds to a Boolean decision tree (strategy) computing $f$.

If  $f$ is given by its truth table, the
SBFE Problem can be exactly solved in time polynomial in the
size of the truth table, using dynamic programming, as in~\cite{guijarroRaghavantruthtable,nijssenFromontoptimal}.  
The following algorithm solves
the SBFE problem with an
approximation factor of $n$ for any function $f$,
even under arbitrary distributions:
Test the variables in increasing cost order (cf.~\cite{kaplanMansour-Stoc05}).
We thus consider a factor of $n$ approximation to be trivial.

We now review related work.  
There is a well-known algorithm that exactly solves the
SBFE problem for disjunctions:
test the $x_i$ in increasing order
of ratio $c_i/p_i$ (see, e.g., \cite{GAREY73}).
A symmetric algorithm works for conjunctions.
There is also a poly-time exact algorithm for
evaluating a $k$-of-$n$ function
(i.e., a function that evaluates to 1 iff at least $k$ of the 
$x_i$ are equal to 1)
~\cite{Salloum79,BenDov81,Salloum84, Chang}. 
There is a poly-time exact algorithm 
for evaluating a read-once DNF formula $f$, but the complexity of the problem is
open when $f$ is a general read-once formula~\cite{borosSequential, Greiner02,Greiner06}.
The SBFE problem is NP-hard for linear threshold functions~\cite{heuristicLeastCostCox},
but for the special case
of unit costs and uniform distribution, 
testing the variables in decreasing order of the magnitude of 
their coefficients is optimal~\cite{unluyurtBorosDoubleRegular,fiatPechyony}.
A survey by \"{U}nl\"{u}yurt \cite{unluyurtReview} covers other results on
exactly solving the SBFE problem.

There is a {\em sample} version of the evaluation
problem, where the input is a sample of size $m$ of
$f$ (i.e, a set of $m$ pairs $(x,f(x))$),
and the problem is to build a decision tree that computes $f$ correctly
on the $x$ in the sample that minimizes the average
cost of evaluation over the sample.
Golovin et al.\ and Bellala et al.\ developed $O(\log m)$ approximation algorithms for arbitrary $f$~\cite{golovinKrauseRayNIPS,
bellalaScottIT2012}, and
there is a 4-approximation algorithm 
when $f$ is a conjunction~\cite{barnoyMinSum,feigeMinSum,MBMW05}.
Moshkov and Chikalov proved a related bound
in terms of a combinatorial measure of the sample~\cite{MoshkovChikalov}.
Moshkov gave an 
$O(\log m)$-algorithm for a worst-case cost variant of this problem~\cite{moshkov}.

A number of non-adaptive
versions of standard and submodular set cover
have been studied.
For example,
Iwata and Nagano~\cite{DBLP:conf/focs/IwataN09} studied
the ``submodular cost set cover'' problem,
where  the cost of the cover
is a submodular function that depends on which subsets are in the cover.
Beraldi and Ruszczynski
addressed a set cover problem
where the set of elements covered by each input subset is
a random variable, and full coverage must be achieved
with a certain probability~\cite{probabilisticSetCover}.

Kaplan et al.\ gave their $O(\log kd)$
approximation factor for monotone CDNF 
(and unit costs, uniform distribution) 
in terms of the expected certificate cost,
rather than the expected cost of the optimal strategy.
The gap between expected certificate cost
and expected cost of the optimal strategy can be large: e.g., for
disjunction evaluation, with unit costs, where Prob$[x_i = 1]$ is $1/(i+1)$,
the first measure is constant, while the second is $\Omega(\log n)$.

Kaplan et al.\ also considered the problem of minimizing the
expected cost of evaluating a Boolean function $f$ 
with respect to a given {\em arbitrary} 
probability distribution,
where the distribution
is given by a conditional probability oracle~\cite{kaplanMansour-Stoc05}.
In the work of Kaplan et al., the goal of evaluation 
differs slightly from ours in that they require the evaluation strategy
to output an ``explanation'' of the function 
value upon termination.
They give as an example the case of evaluating a DNF
that is identically true; 
they require testing of the variables in one term of the DNF in order to output that term as a certificate.  In contrast, under our definitions,
the optimal strategy for evaluating an identically true DNF formula
is a zero-cost one that simply outputs ``true'' and performs no tests.

Charikar et al.~\cite{DBLP:journals/jcss/CharikarFGKRS02} considered
the problem of minimizing the worst-case
ratio between the cost of evaluating $f$ on an input $x$,
and the minimum cost of a certificate contained in $x$.
There are also papers on building {\em identification}
trees of minimum average cost, given $S \subseteq \{0,1\}^n$,
but that problem is fundamentally different than function evaluation
because each $x \in S$ must have its own leaf (cf. \cite{adlerHeeringa}).

We note that there is a connection between the linear threshold evaluation
problem and the {\em Min-Knapsack} problem.
In Min-Knapsack, you are given a set of
items with values $a_i \geq 0$ and
weights $c_i \geq 0$, and the goal is to select a subset of the items
to put in the knapsack
such that the total value of the items is at least $\theta$,
and the total weight is minimized.
We can therefore
solve Min-Knapsack by simulating the above 
algorithm on the linear threshold formula $\sum_{i=1}^n a_i x_i \geq \theta$,
giving the value 1 as the result of each test.
It is easy to modify the above analysis to show that in this
case the ratio $\alpha$ is at most 2, because $C_0$ is empty.
We thus have a combinatorial 2-approximation algorithm for 
Min-Knapsack, based
on ADG. (In fact, the deterministic Dual Greedy algorithm
of Fujito would be sufficient here, since the outcomes of the tests
are predetermined.)
There are several previous combinatorial and non-combintorial 2-approximation algorithms
for Min-Knapsack, and the problem also has a PTAS
(~\cite{carretalMinKnapsack}, cf.~\cite{hanMakino}).

Han and Makino considered an on-line version of the Min-Knapsack
where the items are given one-by-one over time~\cite{hanMakino}.
There is also previous work on the
``stochastic knapsack'' problem, but that work
concerns the standard (max) knapsack problem, not Min-Knapsack.

\section{Preliminaries}

{\bf \noindent Basic notation and definitions.}
A table of notation used in this paper is provided in Appendix \ref{append:tablenotation}.

A partial assignment is a vector $b \in \{0,1,*\}^n$.
We view $b$ as an assignment to variables $x_1, \ldots, x_n$. For partial assignment $b$, we use $dom(b)$ to denote the 
set $\{x_i | b_i \neq *\}$.
We will use $b \in \{0,1\}^n$
to represent the outcomes of binary tests,
where for $l \in \{0,1\}$,
$b_i=l$
indicates that test $i$ was performed and had outcome $l$,
and $b_i=*$ indicates that test $i$ was not performed.

For partial assignments $a,b \in \{0,1,*\}^n$, 
$a$ is an {\em extension} of $b$, written $a \sim b$,
if $a_i = b_i$ for all $b_i \neq *$.
We also say that $b$ is {\em contained} in $a$.
Given Boolean function $f:\{0,1\}^n \rightarrow \{0,1\}$,
a partial assignment $b \in \{0,1,*\}^n$
is a 0-certificate (1-certificate) of $f$ if
 $f(a) = 0$ ($f(a)=1$) for all $a$ such that $a \sim b$.
Given a cost vector $c = (c_1, \ldots, c_n)$,
the cost of a certificate $b$ is $\sum_{j:b_j \neq *} c_j$.

Let $N = \{1, \ldots, n\}$.
In what follows, we assume that utility functions are integer-valued.
In the context of standard work on submodularity,
a {\em utility function} is a function
$g:2^N \rightarrow \mathbb{Z}_{\geq 0}$.
Given $S \subseteq N$ and $j \in N$, $g_S(j)$ denotes
the quantity $g(S \bigcup \{j\}) - g(S)$.

We will also use the term {\em utility function}
to refer to a function $g:\{0,1,*\}^n \rightarrow \mathbb{ Z}_{\geq 0}$ 
defined on partial assignments.
Let
$g:\{0,1,*\}^n \rightarrow \mathbb{ Z}_{\geq 0}$,
be such a utility function, and let
$b \in \{0,1,*\}^n$.
For $S \subseteq N$, let $b^S \in \{0,1,*\}^n$
where $b^S_i = b_i$ for $i \in S$,
and $b^S_i = *$ otherwise.
We define 
$g(S,b) = g(b^S)$.  For
$j \in N$, we define
$g_{S,b}(j) = g(S \bigcup \{j\},b) - g(S,b)$.

For $l \in \{0,1,*\}$.
the quantity $b_{x_i \leftarrow l}$ denotes the partial assignment
that is identical to $b$ except that $b_i = l$. 
We define $g_b(i,l) = g(b_{x_i \leftarrow l}) - g(b)$ if
$b_i = *$, and $g_b(i,l) = 0$ otherwise.
When $b$ represents test outcomes, and test $i$ has not
been performed yet,
$g_b(i,l)$ is the change in utility
that would result from adding test $i$ with outcome $l$.

Given probability vector $p = (p_1, \ldots, p_n)$, 
we use $x \sim D_p$ to denote a random $x$ drawn from
product distribution $D_p$.
For fixed $D_p$, 
$b \in \{0,1,*\}^n$, and $i \in N$,
we use
$E[g_b(i)]$ to denote 
the expected increase in utility that would be
obtained by testing $i$.
In the binary case,
$E[g_b(i)] = p_i g_b(i,1) + (1-p_i)g_b(i,0)$. 
Note that $E[g_b(i)]=0$ if $b_i \neq *$.
For $b \in \{0,1,*\}^n$, 
$p(b) = \prod_{j \in dom(b)} p(b,j)$ where
$p(b,j) = p_j$ if $b_j = 1$, and $p(b,j) = 1-p_j$ otherwise.

Utility  function $g:\{0,1\}^n \rightarrow \mathbb{Z}_{\geq 0}$ is 
{\em monotone} if for $b \in \{0,1,*\}^n$,
$i \in N$ such that $b_i = *$,
and $l \in \{0,1\}$, $g(b_{x_i \leftarrow l}) - g(b) \geq 0$;
in other words, additional information can only increase utility.
Utility function $g$ is {\em submodular}
if for all $b,b' \in \{0,1,*\}^n$
and $l \in \{0,1\}$, 
$g(b_{x_i \leftarrow l}) - g(b)  \geq g(b'_{x_i \leftarrow l}) - g(b') $
whenever $b' \sim b$ and $b_i = b'_i = *$.
In the testing context, if the $n$ test outcomes are predetermined,
submodularity means that the value of a given test (measured by
the increase in utility) will not increase if we delay that test
until later.

\bigskip
{\bf \noindent The Stochastic Submodular Set Cover (SSSC) problem.}
The SSSC problem is similar to
the SBFE problem, except that
the goal is to achieve a cover.  
Let ${\cal O} = \{0,1,\ldots, k-1\}$
be a finite set of $k$ states,
where $k \geq 2$ and $* \not\in {\cal O}$.
In what follows,
we assume the binary case where $k = 2$, although
we will briefly mention extensions to the
$k$-ary case where $k > 2$.
\footnote{To simplify the exposition, we define the SSSC Problem in terms of integer valued utility 
functions.}

In the (binary) SSSC problem,
the input consists of the set $N$, a
cost vector $(c_1, \ldots, c_n)$, where each $c_j \geq 0$, 
a probability vector $p = (p_1, \ldots, p_n)$
where $p \in [0,1]^n$,
an integer $Q \geq 0$,
and a utility
function $g:({\cal O} \bigcup \{*\})^n \rightarrow \mathbb{ Z}_{\geq 0}$.  
Further, 
$g(x) = 0$ if $x$ is the vector that is all $*$'s, and
$g(x) = Q$ if $x \in {\cal O}^n$.
We call $Q$ the {\em goal utility}.
We say that $b \in ({\cal O} \bigcup \{*\})^n$
is a {\em cover} 
if $g(b) = Q$.  The {\em cost} of cover $b$ is $\sum_{j:b_j \neq *} c_j$.

Each {\em item} $j \in N$ has a state $x_j \in {\cal O}$.
We sequentially choose items from $N$.
When we choose item $j$,
we observe its state $x_j$ (we ``test'' $j$).
The states of items chosen so far are represented by a partial
assignment $b \in ({\cal O} \bigcup *)^n$.
When $g(b) = Q$, we have a cover, and we can output it.
The goal is to determine the order in which to choose the items, 
while minimizing the expected testing cost with respect to distribution $D_p$.
We assume that an algorithm for this problem will be executed in 
an on-line setting, and that it can be adaptive.

SSSC is a generalization of Submodular Set Cover (SSC),
which is a generalization of the standard 
(weighted) Set Cover problem, which we call Classical Set Cover.
In Classical Set Cover, the input is a finite ground set $X$, a set $F = \{ S_1, \ldots, S_m\}$ where each $S_j \subseteq X$, and
a cost vector $c = (c_1, \ldots, c_m)$ where each $c_j \geq 0$.
The problem is to find a min-cost ``cover'' $F' \subseteq F$
such that $\bigcup_{S_j \in F'} S_j = X$, and the
cost of $F'$ is $\sum_{S_j \in F'} c_j$.
In SSC, the input is a cost vector $c = (c_1, \ldots, c_n)$,
where each $c_j \geq 0$,
and a utility function $g:2^N \rightarrow {\cal Z}_{\geq 0}$ 
such that $g$ is monotone and submodular, 
$g(\emptyset) = 0$,
and $g(N)=Q$.  The goal is 
to find a subset $S \subseteq N$ such that $g(S) = Q$ and 
$\sum_{j \in S} c_j$ is minimized.
SSC can be viewed as a special case of
SSSC in which each $p_j$ is equal to 1.

\bigskip
{\bf \noindent The Adaptive Greedy algorithm for Stochastic Submodular Set Cover.}
The Classical Set Cover problem has a simple greedy approximation algorithm that
chooses the subset with the ``best bang for the buck" -- i.e.,
the subset covering the most
new elements per unit cost.
The generalization of this algorithm to SSC, due to Wolsey, 
chooses the element 
that adds the maximum additional utility per unit cost~\cite{wolsey}.
The Adaptive Greedy algorithm of Golovin and Krause, for the SSSC problem,
is a further generalization.
It chooses 
the element with the maximum {\em expected} increase in
utility per unit cost.
(Golovin and Kruase actually formulated Adaptive Greedy
for use in solving a somewhat more general problem than SSSC, but here
we describe it only as it applies to SSSC.)
We present the pseudocode for Adaptive Greedy 
in Algorithm \ref{alg:greedy}.

Some of the variables used in the pseudocode are not necessary 
for the running of the algorithm, but are useful in its analysis. (In Step \ref{mainstepa}, assume that if $E[g_b(x)] = 0$, the expression evaluates to 0.)

\begin{algorithm}[ h!]
\caption{Adaptive Greedy}
\label{alg:greedy}

\begin{algorithmic}[h!b]
	\STATE $b \gets (*, * \ldots, *)$ 
    \STATE $l \gets 0$, $F^0 \gets \emptyset$
	\WHILE {$b$ is not a solution to SSSC ($f(b) < Q$)}
        \STATE $l \gets l+1$
        \STATE $j_l \gets \argmin_{j \not\in F^{l-1}} \frac{c_j}{E[g_b(j)]}$ 
        \label{mainstepa}
	\STATE $k \gets$ the state of $j_l$
	\COMMENT ``test'' $j_l$
        \STATE $F^l \gets F^{l-1} \bigcup \{j_l\}$ ~~\COMMENT $F^l = dom(b)$
	\STATE $b_{j_l} \gets k$
	\ENDWHILE
	\STATE return $b$
\end{algorithmic}
\end{algorithm}

Golovin and Krause proved that Adaptive Greedy is a
$(\ln Q + 1)$-approximation algorithm, where $Q$ is the goal utility.
We will make repeated use of this bound. 

\section{Function Evaluation and the SSSC Problem}

\subsection{The $Q$-value approach and CDNF Evaluation.}

{\bf \noindent Definition:} Let $f(x_1,\ldots, x_n)$ be a Boolean function.
Let $g:\{0,1,*\}^n \rightarrow \mathbb{ Z}_{\geq 0}$ be a utility
function.  We say that $g$ is {\em assignment feasible
for $f$, with goal value $Q$},  if (1) $g$ is monotone and submodular,
(2) $g(*, *, \ldots, *)=0$, and (3) for $b \in \{0,1,*\}^n$,
$g(b) = Q$ iff $b$ is a 0-certificate or a 1-certificate of $f$.

We will use the following approach to solving SBFE problems, 
which we call the {\em $Q$-value approach}. 
To evaluate $f$,
we construct an assignment feasible  
utility function $g$ for $f$ with goal value $Q$.
We then run Adaptive Greedy on
the resulting SSSC problem.
Because $g(b) = Q$ iff $b$ is either a 0-certificate or a 1-certificate of $f$,
the decision tree that is (implicitly) output by Adaptive Greedy
is a solution to the SBFE problem for $f$.
By the bound on Adaptive Greedy, this solution
is within a factor of $(\ln{Q} + 1)$ of optimal.

The challenge in using the above approach is in constructing $g$.
Not only must $g$ be assignment feasible, but 
$Q$ should be subexponential, to obtain
a good approximation bound.
We will use the following lemma, due to Guillory and Bilmes, in our construction of $g$.  

\begin{lemma}\mbox{\cite{guillorybilmesICML11}}
\label{combinegoals}
Let $g_0:\{0,1,*\}^n \rightarrow \mathbb{Z}_{\geq 0}$,` 
$g_1:\{0,1,*\}^n \rightarrow \mathbb{Z}_{\geq 0}$, 
and $Q_0, Q_1 \in \mathbb{Z}_{\geq 0}$ be such that
$g_0$ and $g_1$ are monotone, submodular utility functions, 
$g_0(*,*, \ldots, *) = g_1(*,*, \ldots, *) = 0$, and
$g_0(a) \leq Q_0$ and
$g_1(a) \leq Q_1$ for all $a \in \{0,1\}^n$.

Let $Q_{\vee} = Q_0Q_1$
and let
$g_{\vee}:\{0,1,*\}^n \rightarrow \mathbb{Z}_{\geq 0}$ be such that
$g_{\vee}(b) = Q_{\vee} - (Q_0-g_0(b))(Q_1-g_1(b))$. 

Let $Q_{\wedge} = Q_0+Q_1$
and let
$g_{\wedge}:\{0,1,*\}^n \rightarrow \mathbb{Z}_{\geq 0}$ be such that
$g_{\wedge}(b) = g_0(b) + g_1(b)$.

Then $g_{\vee}$ and $g_{\wedge}$ are monotone and submodular, 
and $g_{\vee}(*, \ldots, *) = g_{\wedge}(*, \ldots, *) = 0$.
For $b \in \{0,1,*\}^n$,
$g_{\vee}(b) = Q_{\vee}$ iff 
$g_0(b) = Q_0$ or $g_1(b) = Q_1$, or both.
Further,
$g_{\wedge}(b) = Q_{\wedge}$ iff 
$g_0(b) = Q_0$ and $g_1(b) = Q_1$.
\end{lemma}

Using the $Q$-value approach, it is easy to obtain an algorithm
for evaluating CDNF formulas.
A CDNF formula for Boolean function $f$ is a pair
$(\phi_0,\phi_1)$ where $\phi_0$ and $\phi_1$ are CNF and DNF formulas
for $f$, respectively.

\begin{theorem}
There is a polynomial-time $O(\log kd)$-approximation algorithm solving the SBFE problem for CDNF formulas, 
where $k$ is the number of clauses in the  CNF, and
$d$ is the number of terms in the DNF.
\end{theorem}
\begin{proof}
Let $\phi_0$ be the CNF and $\phi_1$ be the DNF.
Let $f$ be the Boolean function defined by these formulas.
Let $k$ and $d$ be, respectively, the number of clauses and terms of $\phi_0$ and $\phi_1$.
Let $g_0:\{0,1,*\}^n \rightarrow \mathbb{Z}_{\geq 0}$ be such that
for $a \in \{0,1,*\}^n$,
$g_0(a)$ is the number of terms of $\phi_1$ set to 0 by $a$ (i.e.
terms with a literal $x_i$ such that $a_i = 0$, or a literal $\neg{x_i}$
such that $a_i = 1$). Similarly, let 
$g_1(a)$ be the number of clauses
of $\phi_0$ set to 1 by $a$.  Clearly, $g_0$ and $g_1$ are monotone
and submodular.
Partial assignment $b$ is a 0-certificate of $f$ 
iff $g_0(b) = d$ and a 1-certificate of $f$ iff $g_1(b) = k$.
Applying the disjunctive
construction of Lemma~\ref{combinegoals} to $g_1$ and $g_0$,
yields a utility function $g$ 
that is assignment feasible for $f$ with goal value $Q = kd$.
Applying Adaptive Greedy and its $(\ln Q + 1)$ bound
yields the theorem.
\end{proof}

Given a decision tree for a Boolean function $f$, a CNF (DNF) for $f$
can be easily computed using the paths to the 0-leaves (1-leaves) 
of the tree.
Thus the above theorem gives an $O(\ln t)$ approximation algorithm
for evaluating decision trees, where $t$ is the number of leaves.

\subsection{Linear threshold evaluation via the $Q$-value approach}
\label{sec:thresholdQ}
A linear threshold formula with integer coefficients has the form
$\sum_{i=1}^n a_i x_i \geq \theta$ where the $a_i$ and $\theta$ are integers.
It represents the function $f:\{0,1\}^n \rightarrow \{0,1\}$
such that $f(x) = 1$ if
$\sum_{i=1}^n a_i x_i \geq \theta$, and $f(x) = 0$ otherwise.
We show how to use the $Q$-value approach to 
obtain an algorithm solving the SBFE problem for linear threshold formulas with
integer coefficients.  The  algorithm achieves an
$O(\log D)$-approximation, 
$D = \sum_{i=1}^n |a_i|$.
This algorithm, like the CDNF algorithm, works by
reducing the evaluation problem to an SSSC problem.
However,
the CDNF algorithm reduces the evaluation problem to a stochastic version
of {\em Classical} Set Cover problem (each $x_i$ covers
one subset of the (term,clause) pairs when it equals 1, and another
when it equals 0).  Here there is no associated
Classical Set Cover problem.

Let $h(x) = (\sum_{i=1}^n a_i x_i) - \theta$.
For $b \in \{0,1,*\}$,
let $min(b) = \min\{h(b'):b' \in \{0,1\}^n$ and $b' \sim b\}$ 
and
let $max(b) = \max\{h(b'):b' \in \{0,1\}^n$ and $b' \sim b\}$. 
Thus $min(b) = (\sum_{j:b_j \neq *} a_j b_j) + (\sum_{i:a_i < 0,b_i=*} a_i) - \theta$, 
$max(b) = (\sum_{j:b_j \neq *} a_j b_j) + (\sum_{i:a_i > 0,b_i=*} a_i) - \theta$, 
and each can
be calculated in linear time.
Let $R_{min} = min(*, \ldots, *)$
and $R_{max} = max(*, \ldots, *)$.
If $R_{min} \geq 0$ or $R_{max} < 0$, $f$ is constant and
no testing is needed.  Suppose this is not the case.

Let $Q_1 = -R_{min}$ and
let submodular utility
function $g_1$ be such that
$g_1(b) = min\{-R_{min},min(b) - R_{min}\}$,
Intuitively,
$Q_1 - g_1(b)$ is the number of different values
of $h$ that can be induced by extensions $b'$ of $b$
such that $f(b') = 0$.
Similarly, define 
$g_0(b) = min\{R_{max}+1,R_{max}-max(b)\}$
and $Q_0 = R_{max}+1$.
Thus $b$ is a 1-certificate of $f$ iff $g_1(b) = Q_1$,
and a 0-certificate iff $g_0(b) = Q_0$. 

We apply the disjunctive construction
of Lemma~\ref{combinegoals} to construct
$g(b) = Q - (Q_1 - g_1(b))(Q_0 - g_0(b))$,
which is an assignment feasible utility function for $f$
with goal value $Q = Q_1 Q_0$.
Finally, we obtain an $O(\log D)$ approximation bound by
applying the $(\ln Q + 1)$ bound on Adaptive Greedy.

The quantity $D$ 
can be exponential in $n$, the number of
variables.   One might hope to obtain a better approximation
factor, still using the $Q$-value
approach, by designing a more clever assignment-feasible utility function
with a much lower goal-value $Q$.
However, in the next section we show
that this is not possible. 
Achieving a 3-approximation for this problem, as we do 
in Section~\ref{sec:adg}, requires a different approach.

\subsection{Limitations of the Q-value approach}
\label{subsec:limitations}
The $Q$-value approach depends on
finding an assignment feasible utility function $g$ for $f$.
We first demonstrate that a generic such $g$ exists for
all Boolean functions $f$.
Let $Q_0 = |\{a \in \{0,1\}^n | f(a) = 0\}|$ and $Q_1 = |\{a \in \{0,1\}^n | f(a) = 1\}$.
For partial assignment $b$, let $g_0(b) = Q_0 - |\{ a \in \{0,1\}^n| a \sim b, f(a) = 0\}|$
with goal value $Q_0$, and let
$g_1(b) = Q_1 - |\{a \in \{0,1\}^n | a \sim b, f(a) = 1\}|$ with goal value $Q_1$.
Then $g_0, Q_0, g_1$ and $Q_1$  obey the properties of 
Lemma~\ref{combinegoals}.  Apply the disjunctive construction in that lemma,
and let $g$ be the resulting utility function.
Then $g$ is assigment feasible for $f$ with goal value $Q = Q_1Q_0$.
In fact, this $g$ is precisely the utility function
that would be constructed by the approximation algorithm of Golovin et al.
for computing a consistent decision tree of min-expected cost with
respect to a sample,
if we take the sample to be the set of
all $2^n$ entries $(x,f(x))$ in the truth table of $f$~\cite{golovinKrauseRayNIPS}.
The goal value $Q$ of this $g$ is
$2^{\theta(n)}$, so in this case the bound
for Adaptive Greedy, $(\ln Q + 1)$, is linear in $n$.

Since we want a sublinear approximation factor, 
we would instead like to construct an 
assignment-feasible utility function for $f$ 
whose $Q$ is sub-exponential in $n$.
However, we now show this is impossible
even for some simple Boolean functions $f$.
We begin by introducing the following
combinatorial measure of a Boolean function,
which we call its $Q$-value.

\vspace{6pt}
{\bf \noindent Definition:}  The $Q$-value of a Boolean function 
$f:\{0,1\}^n \rightarrow \{0,1\}$
is the minimum integer $Q$ such that there exists a 
assignment feasible utility function $g$ for $f$
with goal value $Q$.

The generic $g$ above
shows that the $Q$-value of every $n$-variable 
Boolean function is at most $2^{O(n)}$.

\begin{lemma}
\label{technical}
Let $f(x_1, \ldots, x_n)$ be a Boolean function, where $n$ is even.
Further, let $f$ be such that for all 
$n' \leq n/2$, and for all 
$b \in \{0,1,*\}^n$,
if $b_i = b_{n/2+i}=*$ for all $i \in \{n'+1, \ldots, n/2\}$,
the following properties hold:
(1) if for all $i \in \{1, \ldots, n'\}$,
exactly one of $b_i$ and $b_{n/2+i}$ is equal to * and the other is equal to 1,
then $b$ is not a 0-certificate or a 1-certificate of $f$ and
(2)  if for all $i \in \{1, \ldots, n'-1\}$,
exactly one of $b_i$ and $b_{n/2+i}$ is equal to * and the other is equal to 1,
and $b_{n'} = b_{n/2+n'} = 1$, then
$b$ contains a 1-certificate of $f$. 
Then the $Q$-value of $f$ is at least $2^{n/2}$. 
\end{lemma}
\begin{proof}  
Let $f$ have the properties specified in the lemma.
For bitstrings $r,s \in \{0,1\}^l$, where $0 \leq l \leq n/2$,
let $d_{r,s} \in \{0,1,*\}^n$ be such that $d_i = r_i$ 
and $d_{n/2+i} = s_i$ for $i \in \{1, \ldots, l\}$,
and $d_i = *$ for all other $i$.
Suppose $g$ is an assignment feasible utility function 
for $f$ with goal value $Q$.
We prove the following claim.
Let $0 \leq l  \leq n/2$.   Then there exists
$r,s \in \{0,1,*\}^l$ 
such that $0 \leq Q - g(d_{r,s}) \leq Q/2^l$, and
for all $i \in \{1, \ldots, l\}$, either 
$r_i = 1$ and $s_i = *$, or $r_i = *$ and $s_i = 1$.

We prove the claim by induction on $l$.
It clearly holds for $l=0$.
For the inductive step, assume it holds for $l$.
We show it holds for $l+1$.
Let $r,s \in \{0,1,*\}^l$ be as guaranteed by the
assumption, so
$Q - g(d_{r,s}) \leq n/2^l$.

For $\sigma \in \{0,1,*\}$,
$r\sigma$ denotes the concatenation of bitstring $r$ with $\sigma$,
and similarly for $s\sigma$.
By the conditions on $f$ given in the lemma,
$d_{r,s}$ is not a 0 or 1-certificate of $f$.
However, $d_{r1,s1}$ is a 1-certificate of $f$ and
so $g(d_{r1,s1}) = Q$.
If $Q - g(d_{r1,s*}) \leq Q/2^{l+1}$, then
the claim holds for $l+1$, because
$r1,s*$ have the necessary properties.
Suppose $Q - g(d_{r1,s*}) > Q/2^{l+1}$.
Then, because $g(d_{r1,s1}) = Q$,
$g(d_{r1,s1}) - g(d_{r1,s*}) > Q/2^{l+1}$.
Note that $d_{r1,s1}$ is the extension of  $d_{r1,s*}$
produced by setting $d_{n/2+l+1}$ to $1$.
Similarly, $d_{r*,s1}$ is the extension of $d_{r*,s*}$
produced by setting $d_{n/2+l+1}$ to $1$.
Therefore, by the submodularity of $g$, 
$g(d_{r*,s1})-g(d_{r*,s*}) \geq g(d_{r1,s1}) - g(d_{r1,s*})$,
and thus
$g(d_{r*,s1})-g(d_{r*,s*}) \geq Q/2^{l+1}$.

Let $A = g(d_{r*,s1}) - g(d_{r*,s*})$ and
$B = Q - g(d_{r*,s1})$. 
Thus $A \geq Q/2^{l+1}$, and
$A+B = Q-g(d_{r*,s*}) = Q-g(d_{r,s}) \leq Q/2^l$ 
where the last inequality is from the original assumptions on $r$ and $s$.
It follows that $B = Q - g(d_{r*,s1}) \leq Q/2^{l+1}$, 
and the claim holds for $l+1$, because
$r*,s1$ have the necessary properties.

Taking $l = n/2$,
the claim says
there exists $d_{r,s}$
such that $Q - g(d_{r,s}) \leq Q/2^{n/2}$.
Since $g$ is integer-valued, 
$Q \geq 2^{n/2}$.
\end{proof}

The above lemma immediately implies the following theorem.

\begin{theorem}
    \label{thm:negative}
Let $n$ be even. Let $f:\{0,1\}^n \rightarrow \{0,1\}$ 
be the Boolean function represented
by the read-once DNF formula 
$\phi = t_1 \vee t_2 \vee \ldots \vee t_{n/2}$
where each $t_i = x_i x_{n/2+i}$.
The $Q$-value of $f$ is at least $2^{n/2}$.
\end{theorem}

The above theorem shows that the $Q$-value approach
will not yield a good approximation bound for
either read-once DNF formulas 
or for DNF formulas with terms
of length 2.  

In the next theorem, we show that there is a particular linear threshold function whose $Q$-value is at least $2^{n/2}$. It follows that the $Q$-value approach will not yield a good
approximation bound for linear-threshold formulas either.

We note that the function described in the next theorem has been sudied before.
As mentioned in~\cite{Hastad94}, 
there is a lower bound of essentially $2^{n/2}$
on the size of the largest integer coefficients in any
representation of the function as a 
linear threshold formula with integer coefficients.

\begin{theorem}
Let $f(x_1, \ldots, x_n)$ be the function defined for
even $n$, whose value is 1 iff the number represented in binary by bits $x_1 \ldots x_{n/2}$
is strictly less than the number represented in binary by bits $x_{n/2+1},\ldots, x_n$, and 0 otherwise. The $Q$-value of $f$ is at least $2^{n/2}$. \end{theorem}

\begin{proof}
We define a new function: 
$f'(x_1, \ldots, x_n) = f(\neg{x_1}, \ldots, \neg{x_{n/2}}, x_{n/2+1}, \ldots, x_n)$.
That is, $f'(x_1, \ldots, x_n)$ is computed by negating the
assignments to the first $n/2$ variables, and then computing
the value of $f$ on the resulting assignment.
Function $f'$ obeys the conditions of Lemma~\ref{technical},
and so has $Q$-value at least $2^{n/2}$.
Then $f$ also has $Q$-value at least $2^{n/2}$, because
the $Q$-value is not changed by the negation of input variables.
\end{proof}

\vspace{-.5em}

Given the limitations of the $Q$-value approach 
we can ask whether there are good alternatives.  
Our new bound on Adaptive Greedy 
is $O(\log P)$, where $P$ is the maximum amount of utility
gained by testing a single variable $x_i$, so we might
hope to use $P$-value in place of $Q$-value.
However, this does not help much:
testing all $n$ variables yields
utility $Q$, so testing one of them alone must yield utility
at least $Q/n$, implying that $P \geq  Q/n$.  
Another possibility might be to exploit the fact that Golovin and Krause's
bounds on Adaptive Greedy apply to a more general class of
utility functions than the assignment feasible utility functions, but we do not pursue that possibility.
Instead, we give a new algorithm for the SSSC problem.

\section{Adaptive Dual Greedy and a 3-approximation for Linear Threshold Evaluation}
\label{sec:adg}

We now present ADG,
our new algorithm for the binary version of the SSSC
problem.  It easily extends to the $k$-ary version, 
where $k > 2$, with no change in the approximation bound.
Like Fujito's Dual Greedy algorithm for the (non-adaptive) SSC problem,
it is based on 
Wolsey's IP for the (deterministic) Submodular
Set Cover Problem.  We present Wolsey's IP in Figure~\ref{WolseyIP}.

\begin{figure}[h!tb]
\caption{Wolsey's IP for submodular set cover}
\label{WolseyIP}
\begin{tabular}{lll}
Min & $\sum_{j \in N} c_jx_j$\\
s.t. \\
& $\sum_{j \in N} g(S \bigcup \{j\}) - g(S) \geq Q-g(S)$ & $\forall S \subseteq N$\\
& $x_j \in \{0,1\}$ & $\forall j \in N$
\end{tabular}
\end{figure}

Wolsey proved that an assignment $x \in \{0,1\}^n$ to the variables in this IP
is feasible 
iff $\{j|x_j = 1\}$ is a cover for the associated Submodular
Set Cover instance, i.e., iff
$g(\{j|x_j = 1\}) = Q$.  We call this Wolsey's property.

In Figure~\ref{LPADG}, we present a new LP, based on Wolsey's IP, which we call LP1.
We use the following notation:
$W = \{w \in \{0,1,*\}^n~|~w_j = *$ for exactly one value of $j\}$.
For $w \in W$, $j(w)$ denotes the $j \in N$ where $w_j = *$.
Further,
$w^{(0)}$ and $w^{(1)}$ denote the extensions of $w$ obtained from $w$ by setting $w_{j(w)}$
to 0 and 1, respectively.
For $a \in \{0,1\}^n$, $j \in N$,
$a^j$ denotes the the
partial assignment obtained from $a$
by setting $a_j$ to $*$. 
We will rely on the following observation, which we call
the Neighbor Property:   Let $T$ be a decision tree solving the SSSC problem.
Given two assignments $a,a' \in \{0,1\}^n$
differing only in bit $j$, either $T$ tests $j$ on both input $a$ and input $a'$,
or on neither.  
\begin{figure}[h!]
\caption{LP1: the Linear Program for Lower Bounding Adaptive Dual Greedy}
\label{LPADG}
\begin{tabular}{lll}
Min & $\sum_{w \in W} c_{j(w)}p(w)x_w$\\
s.t. \\
& $\sum_{j \in N} g_{S,a}(j)x_{a^j} \geq Q-g(S,a)$ & $\forall a \in \{0,1\}^n, S \subseteq N$\\
& $x_w \geq 0$ & $\forall w \in W$
\end{tabular}
\end{figure}
\vspace{-.5em}
\begin{lemma}
\label{lowerbound}
The optimal value of LP1 lower bounds the 
expected cost of an optimal decision tree $T$ for the SSSC instance on $g$,
$p$, and $c$.
\end{lemma}
\begin{proof}
Let $X$ be the assignment to the variables $x_w$ in the LP
such that $x_w = 1$ if
$T$ tests $j$ on both assignments extending $w$,
and $x_w = 0$ otherwise.
With respect to $X$,
the expected cost of $T$ equals $\sum_{a \in \{0,1\}^n} \sum_j c_j x_{a^j} p(a)$.
This equals the value of the objective function, because for
$a, a' \in \{0,1\}^n$ differing only in bit $j$, $P(a) + P(a') = P(a^j)$.
Finally,
for any fixed $a \in \{0,1\}^n$, the subset of constraints 
involving $a$, one for each $S \subseteq N$,
is precisely the set of constraints of Wolsey's IP, if we take the utility
function to be $g_a$ such that $g_a(S) = g(S,a)$.
Since $T$ produces a cover for every $a$, by Wolsey's property, the constraints 
of LP1 involving $a$ are satisfied.
Thus $X$ is a feasible solution to LP1, and
the optimal value of the LP is at most the expected cost of the optimal tree.
\end{proof}
\vspace{-.5em}
We present the pseudocode for ADG
in Algorithm \ref{alg:dualgreedy}. (In Step \ref{mainstep}, assume that if 
$E[g_b(x)] = 0$, the expression evaluates to 0.)
Its main loop is analogous to the main loop in Fujito's Dual Greedy algorithm, except that ADG uses expected
increases in utility, instead of known, deterministic increases
in utility and 
the results of the tests performed on the items already in the cover.
The quantity in Step 5 of ADG relies only on the outcomes of completed tests,
so ADG can be executed in our stochastic setting.
\begin{figure}[h!]
\caption{LP2: the Linear Program for Adaptive Dual Greedy}
\label{LP2}
\begin{tabular}{lll}
Max & $\sum_{a \in \{0,1\}^n} \sum_{S \subseteq N} p(a)~(g(N,a)-g(S,a))~y_{S,a}$ \\
s.t. \\
& $\sum_{S \subseteq N} (1-p_{j(w)})~g_{S,w^{(0)}}(j)~y_{S,w^{(0)}}  + \sum_{S \subseteq N} p_{j(w)}~g_{S,w^{(1)}}(j)~y_{S,w^{(1)}}\leq c_j$ & $\forall w \in W$\\
& $y_{S,a} \geq 0$ & $\forall S \subseteq N, a \in \{0,1\}^n$
\end{tabular}
\end{figure}
      \begin{algorithm}[H]
\caption{Adaptive Dual Greedy}
\label{alg:dualgreedy}

\begin{algorithmic}[h!]

	\STATE $b \gets (*, * \ldots, *)$, $y_S \gets 0$ for all $S \subseteq N$
        \STATE $F^0 = \emptyset$, $l \gets 0$
	\WHILE {$b$ is not a solution to SSSC ($g(b) < Q$)}
        \STATE $l \gets l+1$
	\STATE $j_l \gets \argmin_{j \not\in F^{l-1}} \frac{c_j - \sum_{S:y_S \neq 0} (E[g_{S,b}(j)]) y_S}{E[g_b(j)]}$ 
        \label{mainstep}
        \STATE $y_{F^{l-1}} \gets \frac{c_{j_l} - \sum_{S:y_S \neq 0} (E[g_{S,b}(j)]) y_S}{E{[g_b(j)]}}$	
	\STATE $k \gets$ the state of $j_l$ ~~~~~~~~~\COMMENT ~~``test'' $j_l$
        \STATE $F^l \gets F^{l-1} \bigcup \{j_l\}$ ~~~~~~~~\COMMENT ~~$F^l = dom(b)$
	\STATE $b_{j_l} \gets k$
	\ENDWHILE
	\STATE return $b$
\end{algorithmic}
\end{algorithm}


We now analyze Adaptive Dual Greedy.  In the LP in Figure~\ref{LPADG},
there is a constraint for each $a,S$ pair.  Multiply both sides of
such constraints by $p(a)$,  to form an equivalent LP. Take the dual
of the result, and divide both sides of each constraint by $p(w)$.
(Note that $p(a)/p(w) = p(a,j(w))$.)
We give the resulting LP, which we call LP2, in Figure~\ref{LP2}.
The variables in it are $y_{S,a}$, where $S \subseteq N$
and $a \in \{0,1\}^n$.

Consider running ADG on an input $a \in \{0,1\}^n$.
Because $g(a) = Q$, ADG is guaranteed to terminate with an output
$b$ such that $g(b) = Q$.
Let $C(a) = dom(b)$.
That is, $C(a)$ is the set of items that ADG tests and inserts into the cover it constructs for $a$.
We will sometimes treat $C(a)$ as a sequence of items, ordered by their insertion order.
ADG constructs
an assignment to the variables $y_S$ (one for each $S \subseteq N$)
when it is run on input $a$.
Let $Y$ be the assignment to the variables $y_{S,a}$ of LP2, 
such that $Y_{S,a}$ is the value of ADG variable $y_S$
at the end of running ADG on input $a$.


We now show that $Y$ is a feasible solution to LP2 and that
for each $a$, and each $j \in C(a)$, $Y$ makes the constraint for $d=a^j$ tight.
For $w \in W$, let $h'_w(y)$ denote the function of
the variables $y_{S,a}$ computed in the left hand side of the constraint for $w$ in LP2.

\begin{lemma} \label{mainlemma}
For every $a \in \{0,1\}^n$, $j \in N$,\\
(1) $h'_{a^j}(Y) = c_j$ if $j \in C(a)$, and \\
(2) $h'_{a^j}(Y) \leq c_j$ if $j \not\in C(a)$.
\end{lemma}

\begin{proof}
Assignment $Y$ assigns non-zero values only to variables $y_{S,a}$
where $S$ is a prefix of sequence $C(a)$.

For $t \in N$, let $Y^t$ denote the assignment to the $y_{S,a}$ variables such that
$y_{S,a}$ equals the value
of variable $y_S$ at the end of iteration $t$ of the loop in ADG,
when ADG is run on input $a$.
(If ADG terminates before iteration $t$, $y_{S,a}$
equals the final value of $y_S$).
Let $Y^0$ be the all 0's assignment. 
We begin by showing that for all $t$ and $a$,
$h'_{a^j}(Y^t) = \sum_{S \subseteq N} E[g_{a(S)}(j)] Y^t_{S,a}$. Recall that 
$a(S) \in \{0,1,*\}^n$ such that $\forall i \in S$, $a(S)_i = a_i$,  and $\forall j \notin S$, $a(S)_j = \ast$. 

Consider running ADG on $w^{(0)}$ and $w^{(1)}$.
Since ADG corresponds to a decision tree, the Neighbor Property holds. Then, if $j$ is never tested on $w^{(0)}$, it is never tested on $w^{(1)}$.
and $Y^t_{S,w^{(0)}} = Y^t_{S,w^{(1)}}$ for all $S,t$.
Thus $h'_{w}(Y^t) = \sum_{S \subseteq N} (p_j g_{S,w^{(1)}}(j)) + (1-p_j) g_{S,w^{(0)}}(j))Y^t_{S,w^{(1)}} = \sum_{S \subseteq N} E[g_{S,w}(j)] Y^t_{S,w^{(1)}}$
for all $t$.

Now suppose that $j$ is tested  in iteration $\hat{t}$ on input $w^{(1)}$, and hence on input $w^{(0)}$.
For $t \leq \hat{t}$, $Y^t_{S,w^{(1)}} = Y^t_{S,w^{(0)}}$ for all $S$.
This is not the case for $t > \hat{t}$.
However, in iterations $t > t'$, $j$ is already part of the cover, so
ADG assigns values only to variables $y_S$ where $j \in S$.
For such $S$, $g_{S,w^{(1)}}(j) = 0$.
Thus in this case also, 
$h'_{w}(Y^t) = \sum_{S \subseteq N} E[g_{S,w}(j)] Y^t_{S,w^{(1)}}$
for all $t$.

It is now easy to show by induction on $t$ that the 
the two properties of the lemma hold for
every $Y^t$, and hence for $Y$.
They hold for $Y^0$.
Assume they hold for $Y^t$.
Again consider assignments $w^{(1)}$ and $w^{(0)}$.
If $j$ was tested on $w^{(1)}$ and $w^{(0)}$ in some iteration $t' < t+1$, 
then $h'_{a^j}(Y^t) = h'_{a^j}(Y^{t+1})$ by the arguments above. 
If $j$ is tested in iteration $t+1$ on both inputs,
then the value assigned to $y_{F^{l-1}}$ by ADG
on $w^{(1)}$ (and $w^{(0)}$) equals 
$(c_j -  h'_{w}(Y^t))/E[g_{F^{l-1},w}(j)]$, and thus
$h'_{w}(Y^{t+1}) = c_j$.
If $j$ is not tested in iteration $t+1$, and was
not tested earlier,
the inductive assumption and the
greedy choice criterion ensure that
$h'_{w}(Y^{t+1}) \leq c_j$. 
\end{proof}

The expected cost of the cover produced by ADG on a random input $a$
is $\sum\nolimits_{a \in \{0,1\}^n}\sum\nolimits_{j \in C(a)}p(a)c_j$.

\begin{lemma} \label{lem:tightconstraints}
$\sum\nolimits_{a \in \{0,1\}^n} \sum\nolimits_{j \in C(a)} p(a) c_j =
\sum\nolimits_{a \in \{0,1\}^n} \sum\nolimits_{S \subseteq N} \sum\nolimits_{j:j \in C(a)} p(a) g_{S,a}(j) Y_{S,a}$
\end{lemma}
\begin{proof}
For $j \in N$, let $W^j = \{w \in W|w_j = *\}$. Then
{
\begin{flalign*}
&\sum\nolimits_a\sum\nolimits_{j:j \in C(a)} p(a) c_j\\
&= \sum\nolimits_j\sum\nolimits_{a:j \in C(a)} p(a) c_j   \hspace{19em} \text {switching the order of summation}\\
&= \sum\nolimits_j (~\sum\nolimits_{w \in W_j:j \in C(w^{(1)})}~p(w^{(1)})c_j + 
\sum\nolimits_{w \in W_j:j \in C(w^{(0)})} p(w^{(0)})c_j) \hspace{1.5em}
\text{grouping assignments by the value of bit $j$}\\
&= \sum\nolimits_j  (~\sum\nolimits_{w \in W_j:j \in C(w^{(1)})}~p(w^{(1)})c_j + p(w^{(0)})c_j)\hspace{1.5em}
\text{because $j \in C(w^{(1)})$ iff $j \in C(w^{(0)})$ by the Neighbor Property}\\
&= \sum\nolimits_j ~\sum\nolimits_{w \in W_j:j \in C(w^{(1)})}~p(w)c_j\\
&= \sum\nolimits_j ~\sum\nolimits_{w \in W_j:j \in C(w^{(1)})}~p(w)h'_w(Y)\hspace{23em} 
\text{by Lemma~\ref{mainlemma}}\\
&= \sum\nolimits_j~\sum\nolimits_{w \in W_j:j \in C(w^{(1)})}~(p(w)\sum\nolimits_S(p_i~g_{S,w^{(1)}}(j)~Y_{S,w^{(1)}} + (1-p_i)~g_{S,w^{(0)}}(j)~Y_{S,w^{(0)}}))\hspace{1.5em}
\text{by the definition of $h'_w$}\\
&= \sum\nolimits_j~\sum\nolimits_{w \in W_j:j \in C(w^{(1)})}~\sum\nolimits_S (p(w^{(1)})~g_{S,w^{(1)}}(j)~Y_{S,w^{(1)}} + p(w^{(0)})g_{S,w^{(0)}}(j)Y_{S,w^{(0)}})\\
&= \sum\nolimits_j [(~\sum\nolimits_{w \in W_j:j \in C(w^{(1)})}~\sum\nolimits_S p(w^{(1)})g_{S,w^{(1)}}(j)Y_{S,w^{(1)}}) + 
(~\sum\nolimits_{w \in W_j:j \in C(w^{(0)})}~\sum\nolimits_S p(w^{(0)})g_{S,w^{(0)}}Y_{S,w^{(0)}})]\\ 
&\text{because $j \in C(w^{(1)})$ iff $j \in C(w^{(0)})$}\\
&= \sum\nolimits_j ~\sum\nolimits_{a:j \in C(a)}~\sum\nolimits_S p(a) g_{S,a}(j)Y_{S,a}\\
&=\sum\nolimits_a\sum\nolimits_S\sum\nolimits_{j:j \in C(a)} p(a) g_{S,a}(j) Y_{S,a}
\end{flalign*}
}
\end{proof}

We now give our approximation bound for ADG.
\vspace{-.5em}
\begin{theorem}
\label{thm:dualg}
Given an instance of SSSC with utility function $g$ and goal value $Q$,
ADG constructs a cover whose expected cost is
no more than 
a factor of $\alpha$ larger than the expected cost of the cover
produced by the optimal strategy, where
$\alpha  = \max  \frac{ \sum_{j \in C(a)} g_{S,a}(j) } { Q - g(S,a)}$,
with the max taken over all $a \in \{0,1\}^n$ and $S \in\mbox{Pref}(C(a))$
such that the denominator is non-zero.
Here $\mbox{Pref}(C(a))$ denotes the set of all prefixes of the cover $C(a)$ that ADG constructs on input $a$.
\end{theorem}
\begin{proof}
By Lemma~\ref{lem:tightconstraints}, the expected cost of the cover
constructed by ADG
is $\sum\nolimits_{a} \sum\nolimits_{S} \sum\nolimits_{j:j \in C(a)} p(a) g_{S,a}(j) Y_{S,a}$.
The value of the objective function of LP2 on $Y$ is
$\sum_{a} \sum_{S} p(a)~(Q-g(S,a))~Y_{S,a}$.
For any $a$, $Y_{S,a}$ is non-zero if $S \in \mbox{Pref}(C(a))$.
Comparing the coefficients of $Y_{S,a}$
in these two expressions
implies that the value of the objective function on $Y$
is at most 
$\max  \frac{ \sum_{j \in C(a)} g_{S,a}(j) } { Q - g(S,a)}$
times the expected cost of the cover.
The theorem follows by Lemma~\ref{lowerbound} and
weak duality.
\end{proof}


\begin{theorem} 
\label{thm:thresholdDual}
There is a polynomial-time 3-approximation algorithm solving SBFE problem for linear threshold formulas with
integer coefficients.
\end{theorem}
\begin{proof}
We modify the linear threshold evaluation algorithm from Section
~\ref{sec:thresholdQ}, substituting
ADG for Adaptive Greedy.
By Theorem~\ref{thm:dualg}, the resulting algorithm is within
a factor of $\alpha$ of optimal.
We now show that $\alpha \leq 3$ in this case.

Fix $x$ and consider the run of ADG on $x$.
Let $T$ be the number of loop iterations.
So $C(x) = j_1, \ldots, j_T$ is the sequence of tested items,
and $F^t = \{j_1, \ldots, j_t\}$.
Assume first that $f(x)=1$. 
Let $F = F^0 = \emptyset$, and consider the ratio
$\frac{ \sum_{j \in C(x)} g_{F,x}(j) } { Q - g(F,x)}$.

We use the definitions and utility functions
from the algorithm in Section~\ref{sec:thresholdQ}.
Assume without loss of generality that neither $g_0$ nor $g_1$
is identically 0.

Let $A = -R_{min}$
and let $B = R_{max} + 1$.
Thus $Q - g(\emptyset,x) = AB$.
Let $C_1$ be the set of items $j_l$ in $C(x)$ such that 
either $x_{j_l} = 1$
and $a_{j_l} \geq 0$ or
$x_{j_l} = 0$
and $a_{j_l} < 0$.
Similarly, let $C_0$ be the set of items $j_l$ in $C(x)$,
such that 
either $x_{j_l} = 0$
and $a_{j_l} \geq 0$ or
$x_{j_l} = 1$
and $a_{j_l} < 0$.

Testing stops as soon as the goal utility is reached. Since
$f(x)=1$, this means
testing on $x$ stops 
when $b$ satisfies $g_1(b) = Q_1$, or equivalently, 
$b$ is a 1-certificate of $f$. 
Thus the last tested item, $j_T$, is in $C_1$.
Further, the sum of the $a_{j_l} x_{j_l}$
over all $j_l \in C_1(x)$, excluding $j_T$,
is less than $-R_{min}$, while the sum including
$j_T$ is greater than or equal to $-R_{min}$.
By the definition of utility function $g$, 
$\sum_{j_l \in C_1:j_l \neq j_T} g_{\emptyset,x}(j_l) < AB$.
The maximum possible value of $g_{\emptyset,x}(j_T)$ is $AB$.
Therefore,
$\sum_{j_l \in C_1} g_{\emptyset,x}(j_l) < 2AB$.

Since $x$ does not contain both a 0-certificate and
a 1-certificate of $f$,
the sum of the $a_{j_l} x_{j_l}$ over all
$j_l \in C_0(x)$ is strictly less than $R_{max}$.
Thus by the definition of $g$,
$\sum_{j_l \in C_0} g_{\emptyset,x}(j_l) < AB$.
Summing over all $j_l \in C(x)$, we get that
$\sum_{j_l \in C(x)} g_{\emptyset,x}(j_l) < 3AB$.
Therefore,
$\frac{ \sum_{j \in C(x)} g_{F,x}(j) } { Q - g(F,x)}$ $< 3$,
because for $F = \emptyset$,
$Q = AB$ and $g(\emptyset,x) = 0$.
A symmetric argument holds when $f(x) = 0$.

It remains to show that the same bound holds when
$F \neq \emptyset$.
We reduce this to the case $F = \emptyset$. 
Once we have tested the variables in $F^t$,
we have an induced linear threshold evaluation problem
on the remaining variables (replacing the tested variables by their values).  
Let
$g'$ and $Q'$ be the utility function and goal value
for the induced problem, as constructed in
the algorithm of Section~\ref{sec:thresholdQ}.
The ratio
$\frac{ \sum_{j \in C(x)} g_{F,x}(j) } { Q - g(F,x)}$ 
is equal to
$\frac{ \sum_{j \in C(x) - F} g'_{\emptyset,x'}(j) } { Q' - g'(\emptyset,x')}$ ,
where $x'$ is $x$ restricted to the elements not in $F$.
By the argument above, this ratio is bounded by 3.
\end{proof}

\section{A new bound for Adaptive Greedy}
\label{append:alternative}

We give a new analysis of the Adaptive Greedy algorithm of Golovin and Krause, 
whose pseudocode we presented in Algorithm 1.
Throughout
this section, we let $g(j) = \max_{l \in \{0,1\}}g_r(j,l)$
where $r = (*, \ldots, *)$. Thus $g(j)$ is the maximum increase in
utility that can be obtained as a result of testing $j$ (since $g$ is
submodular).
We show that the expected cost of the solution computed by
Adaptive Greedy is
within a factor of 
$2(\ln (\max_{i \in N}{g(i)} + 1))$ 
of optimal in the binary case.

In the $k$-ary case, the $2$ in the bound is replaced by $k$.
Note that
$\max_j{g(j)}$ is clearly upper bounded by $Q$, 
and in some instances may be much less than $Q$.
However,
because of the factor of $k$ at the
front of our bound, we cannot say that it is strictly better
than the $(\ln Q + 1)$ bound of Golovin and Krause.
(The bound that is analogous to ours in the non-adaptive case,
proved by Wolsey,
does not have a factor of 2.)

Adaptive Greedy is a natural extension of the Greedy algorithm
for (deterministic) submodular set cover of Wolsey.
We will extend Wolsey's analysis \cite{wolsey}, 
as it was presented by Fujito \cite{fujitoSurvey}.
In our analysis,
we will refer to LP2 defined in Section~\ref{sec:adg},
along with the associated notation for the constraints $h'_{a^j}(y) \leq c_j$.

For $x \in \{0,1\}^n$,
let $T^x$ be the number of iterations of the Adaptive Greedy while loop on input $x$. 
Let $b_x^t$ denote the value of $b$ at the end
of iteration $t$ of the while loop on input $x$,
and let $F^t_x$ denote the value of $F^t$, where $F^t$ is the set of $j$s tested by the end of the $t+1$st iteration. (The $x$ in the notation may be dropped when it is understood implicitly.)  Set $\theta_x^t = \min_{j \notin F^{t-1}} \frac{c_j}{E[g_{b^{t-1}}(j)]}$.

For $j \in N$, 
let $k_j$ be the value of $t$ that maximizes 
($\theta_x^t)(g_{b_x^t}(j,1))$. Similarly, let $l_j$ be the value of $t$ that 
maximizes $(\theta_{x'}^t)(g_{b_{x'}^{t-1}}(j,0))$, where $x'$
is the assignment obtained from $x$ by complementing $x_j$. 
Again, let $r$ denote the assignment $\{*, \ldots, *\}$,
and let
$H_j^1 = H(g_r(j,1))$ and $H_j^0 = H(g_r(j,0))$, where
$H(n)$ denotes the $n$th harmonic number, which is at most $(\ln n + 1)$. 
Let $q_j = 1 - p_j$.

To analyze Adaptive Greedy, we define 
$Y$ to be the assignment to the LP2 variables $y_{S,x}$ 
setting $y_{F^0,x} = \theta_x^1$, $y_{F^t,x} = (\theta_x^{t+1} - \theta_x^t)$ for $t \in \{1 \ldots T^x-1\}$ and $y_{S,x} = 0$ for all other $S$.
We define $Y^x$ to be the restriction of that assignment to variables $y_{S,x}$
for that $x$. Let $q^x(y) = \sum_{S \subset N} g_{S,x}(N-S) y_S^x$.

\begin{lemma}\label{fujitocostgreedy}
The expected cost of the cover
constructed by Adaptive Greedy is at most $E[q^x(Y^x)]$, where the expectation is with respect to $x \sim D_p$.
\end{lemma} 

\begin{proof} 
By the definition of $Y^x$, the proof follows directly from the analysis
of (non-adaptive) Greedy in Theorem 1 of \cite{fujitoSurvey},
by linearity of expectation. \end{proof}

We need to bound the value of $h'_w(Y)$ for each $w \in W$.
We will use the following lemma from Wolsey's analysis.

\begin{lemma}
\label{fujitoseries}
\cite{wolsey} Given two sequences $(\alpha^{(t)})_{t=1}^T$ and $(\beta^{(t)})_{t=0}^{T-1}$, such that both are nonnegative, the former is  monotonically nondecreasing and the latter, monotonically non-increasing, and $\beta^{(t)}$ is a nonnegative integer for any value of $t$, then
\[\alpha^{(1)}\beta^{(0)} + (\alpha^{(2)} - \alpha^{(1)})\beta^{(1)} + \ldots + (\alpha^{(T)} - \alpha^{(T-1)})\beta^{(T-1)}\]
\[\leq (\max_{1 \leq t \leq T} \alpha^{(t)}\beta^{(t-1)}H(\beta^{(0)})). \]
\end{lemma}

\begin{lemma}\label{lem:expectedcostys}
For every $x \in \{0,1\}^n$ and $j \in \{1, \ldots, N\}$, $h'_{x^j}(Y) \leq  c_j2H(\max_{i \in N} g(i))$.
\end{lemma}

\begin{proof} 
By the submodularity of $g$, and the greedy choice criterion used
by Adaptive Greedy,
$\theta_x^1 \leq \theta_x^2 \ldots \leq \theta_x^{T^x}$.
By the submodularity of $g$,
$g_{b_{x}^{0}}(j,0) \geq g_{b_{x}^1}(j,0) \ldots \geq g_{b_{x}^{T^x}}(j,0)$.
Thus Lemma~\ref{fujitoseries}
applies to the non-decreasing sequence
$\theta_x^1, \theta_x^2, \ldots \theta_x^{T^x}$
and the
non-increasing sequence $g_{b_x}^0(j,0), \ldots, g_{b_{x}}^1(j,0), \ldots, g_{b_x}^{T^x-1}(j,0)$.
This also holds if we substitute $(j,1)$ for $(j,0)$ in
the second sequence.

Let $x'$ be the assignment differing from $x$ only in bit $j$.
In the following displayed
equations, we write $k$ and $l$ in place of $k_j$ and $l_j$
to simplify the notation.

\begin{flalign*}
&h'_{x^j}(Y)\\
&= \sum_{S \subseteq N} (p_j g_{S,x}(j) + (1-p_j) g_{S,x'}(j))Y_{S,x} \\
& \hspace{10em}\text {by the Neighbor Property, by the same argument as used in the analysis
of ADG}\\
& = p_j [\theta_x^1g_{b_x^{0}}(j,1) + \Sigma_{i=2}^{T^x} (\theta_x^i - \theta_x^{i-1}) g_{b_x^{i-1}}(j,1)]+q_j [\theta_{x'}^1g_{b_{x'}^{0}}(j,0)+ \Sigma_{i=2}^{T^{x'}} (\theta_{x'}^i - \theta_{x'}^{i-1}) g_{b_{x'}^{i-1}}(j,0)]\\
&\leq p_j [\theta_x^k g_{b_x^{k-1}}(j,1)H_j^1] + q_j [{\theta_{x'}}^{l}g_{b_{x'}^{{l}-1}}(j,0)H_j^0] \hspace{9.5em}\text{  by Lemma \ref{fujitoseries} as indicated above}\\
&\leq p_j [\theta_x^k g_{b_x^{k-1}}(j,1)H_j^1] + q_j [\theta_x^k g_{b_x^{k-1}}(j,0)H_j^0] 
+ p_j [\theta_{x'}^{l} g_{b_{x'}^{l-1}}(j,1)H_j^1] + q_j [\theta_{x'}^{l}g_{b_{x'}^{{l}-1}}(j,0)H_j^0]\\
&\hspace{23.5em}\text{ since this just adds extra non-negative terms}\\
&=\theta_x^kH_j^1[p_jg_{b_x^{k-1}}(j,1) + q_jg_{b_x^{k-1}}(j,0)] 
 + \theta_{x'}^{l}H_j^0[q_jg_{b_{x'}^{l-1}}(j,0) + p_jg_{b_{x'}^{l-1}}(j,1)]\\
&\leq c_jH_j^1 + c_jH_j^0 \hspace{15em}\text{\;due to the greedy choices made by Algorithm \ref{alg:greedy}} \\
&\leq c_j 2H(g(j)))\\
&\leq c_j  2H(\max_i g(i)))\end{flalign*}\end{proof}

\begin{theorem}
\label{thm:greedy}
Given an instance of SSSC with utility function $g$,
Adaptive Greedy constructs a decision tree whose expected cost is
no more than 
a factor of $2(\max_{i \in N}(\ln g(i))+1)$ larger than the expected cost of the cover
produced by the optimal strategy.
\end{theorem}

\begin{proof}
Let $OPT$ be the expected cost of the cover produced by the optimal strategy.
let $AGCOST$ be the expected cost of the cover produced by Adaptive Greedy,
and let $q(y)$ denote the objective function of LP2.
By Lemma~\ref{lowerbound}, the optimal value of LP1 is a lower
bound on OPT.
By Lemma \ref{lem:expectedcostys}, $Z = Y/(2H(\max_i g(i)))$ is
a feasible solution to LP2.
Thus by weak duality, $q(Z) \leq OPT$.
By Lemma~\ref{fujitocostgreedy},
$AGCOST \leq E[q^x(Y^x)]$, and it is easy to see
that 
$E[q^x(Y^x)] = q(Y)$. Since $q(Y) = q(Z)(2H(\max_i g(i)))$,
$AGCOST \leq OPT(2H(\max_i g(i)))$.
\end{proof}

\section{Simultaneous Evaluation and Ranking}
\label{sec:evalranking}

Let $f_1, \ldots, f_m$ be 
(representations of) Boolean functions from a class $C$, such
that each $f_i:\{0,1\}^n \rightarrow \{0,1\}$.
We consider the generalization of the SBFE
problem where instead of determining the value of a single
function $f$ on an input $x$, 
we need to determine the value of all $m$ functions
$f_i$ on the same input $x$.  

The $Q$-value approach can be easily extended to this problem
by constructing utility functions for each of the $f_i$,
and combining them using the conjuctive construction in
Lemma~\ref{combinegoals}.  
The algorithm of 
Golovin and Krause for simultaneous evaluation of OR formulas
follows this approach~\cite{golovinKrause}
(Liu et al. presented a similar algorithm earlier,
using a different analysis~\cite{liuetal}.)
We can also modify the approach by calculating a bound based on $P$-value, or
using ADG instead of Adaptive Greedy.
We thus obtain the following theorem, where
$\sum_{i=1}^n a_{k_i} x_i \leq \theta_k$. 
is the $k$th threshold formula. 

\begin{theorem}
\label{thm:simthresh}
There is a polynomial-time algorithm for solving the simultaneous evaluation of
linear threshold formulas problem which produces a solution that is within a factor of $O(\log mD_{avg})$
of optimal where $D_{avg}$, is the average, over $k \in \{1,\ldots, m\}$,
of $\sum_{i=1}^n |a_{k_i}|$.
In the special case of OR formulas,
where each variable appears in at most $r$ of them, 
the algorithm achieves an approximation factor of
$2(\ln(\beta_{max} r)+1)$, where $\beta_{max}$ 
is the maximum number of variables in any of the OR formulas.

There is also a polynomial-time algorithm for solving the simultaneous evaluation of
threshold formulas problem which produces a solution that is within a factor of $D_{max}$ of 
optimal, where $D_{max} = \max_{k \in \{1,\ldots, m\}} \sum_{i=1}^n |a_{k_i}|$. 
\end{theorem}
\begin{proof}
Let $g^{(1)}, \ldots, g^{(m)}$
be the $m$ utility functions that would be constructed if we ran
the algorithm from Section~\ref{sec:thresholdQ} separately on each 
of the $m$ threshold formulas that need to be evaluated.
Let $Q^{(1)}, \ldots, Q^{(m)}$ be the associated goal values.

Using the conjunctive construction from Lemma~\ref{combinegoals},
we construct utility function $g$
such that
$g(b) = \sum_{k=1}^m g^{(k)}(b)$,
and $Q = \sum_{k=1}^m Q^{(k)}$.

To obtain the first algorithm, we
evaluate all the threshold formulas by
running Adaptive Greedy with $g$, goal
value $Q$, and the given $p$, and $c$,
until it outputs a cover $b$.
Given cover $b$, 
it is easy to determine for each $f_k$ whether
$f_k(x)=1$ or $f_k(x) = 0$.

In the algorithm of Section~\ref{sec:thresholdQ}, for each $f_k$, 
the associated $Q_k = O(D_k)$, where $D_k$ is the sum of
the absolute values of the coefficients in $f_k$. 
Since $Q = \sum_k Q_k$, the $O(\log (mD_{avg}))$ bound follows
from the $(\ln Q + 1)$ bound for Adaptive Greedy. 

Suppose each threshold formula is an OR formula.
For $b \in \{0,1,*\}^n$,
$\max_{l \in \{0,1\}} g^{(k)}_b(i,l) = 0$ if $x_j$ does not appear in the $k$th OR formula,
otherwise it is equal to the number of variables in that formula.
The $2(\ln(\beta_{max} r)+1)$ approximation factor then
follows by our bound on Adaptive Greedy in Theorem~\ref{thm:greedy}.

For the second algorithm, we just use ADG instead of Adaptive Greedy
with the same utility function $g$.
By Theorem~\ref{thm:dualg}, the approximation
factor achieved by ADG is
$\max \frac{ \sum_j g_{S,x}(j) } { Q - g(S,x)}$. 

We bound this ratio for $g$.  Let $D_j = \sum_{i=1}^n |a_{j_i}|$.
Let $d \in \{0,1\}^n$ and $S \in F(x)$.
Without loss of generality, assume $S = \{n'+1, \ldots, n\}$.
In the $k$th threshold formula, for $i \geq n'$,
replace $x_i$ with $d_i$.
This induces a new threshold formula on $n-n'$ variables
with threshold 
$\theta_{k,d} = \theta_k-D_{k,d}$
whose coefficients sum to $D_{k,b} =
D_k - \sum_{i=n'}^n a_id_i$.
Let $b$ be the partial assignment such that $b_i = d_i$ 
for $i \geq n'$, and $b_i = *$ otherwise.
If $b$ contains either a 0-certificate or a 1-certificate for $f_k$, then
$Q_k - g^k(S,d) = 0$.

Otherwise,
$Q_k - g(S,d) = (\theta_{k,b})(D_{k,b}-\theta_{k,b}+1)$, and
$\sum_j g_{S,d}^k(j) \leq D_{k,b}\max\{\theta,D_{k,b}-\theta_{k,b}+1\}$. 
It follows that
$\frac{\sum_j g_{S,d}^k(j)}{Q_k - g^k(S,d)} \leq D_{k,b} \leq D_{max}$.

Since this holds for each $k$, 
$\max \frac{ \sum_j g_{S,d}(j) } {Q - g(S,d)} \leq D_{max}$.
\end{proof}

%
For the special case of simultaneous evaluation of OR formulas,
the theorem implies a $\beta$-approximation algorithm,
where $\beta$ is 
the length of the largest OR formula.
This improves the $2\beta$-approximation achieved by the randomized
algorithm of Liu et al.~\cite{liuetal}.

We use a similar approach to solve the
{\em Linear Function Ranking} problem.
In this problem, you are given
a system of linear functions $f_1, \ldots, f_m$,
where for $j \in \{1, \ldots, m\}$, $f_j$  
is $a_{j_1} x_1 + a_{j_2} x_2 + \ldots a_{j_n} x_n$,
and the coefficients $a_{j_i}$ are integers.
You would like to determine 
the sorted order of the values $f_1(x), \ldots, f_m(x)$,
for an initially unknown $x \in \{0,1\}^n$.
(Note that the values of the $f_j(x)$ are not Boolean.)
We consider the problem of finding an optimal testing
strategy for this problem, where 
as usual, $x \sim D_p$, for some probability  vector $p$,
and there is a cost vector $c$ specifying the cost
of testing each variable $x_i$.

Note that there may be more than one correct output for
this problem if there are ties.  So, strictly speaking,
this is not a function evaluation problem.  Nevertheless, we
can still exploit our previous techniques.
For each system of linear equations $f_1, \ldots, f_m$ over $x_1, \ldots, x_n$,
and each $x \in \{0,1\}^n$, 
let $f(x)$ denote the set of permutations 
$\{f_{j_1}, f_{j_2}, \ldots, f_{j_m}\}$
of $f_1, \ldots, f_m$ such that
$f_{j_1}(x) \leq f_{j_2}(x) \leq \ldots \leq f_{j_m}(x)$.
The goal of sorting the $f_j$ 
is to output {\em some} permutation that we know definitively to
be in $f(x)$.
Note that in particular, if e.g., $f_i(x) < f_j(x)$, it may be enough
for us to determine that $f_i(x) \leq f_j(x)$.

\begin{theorem}
\label{thm:fnRanking}
There is an algorithm that solves the Linear Function Ranking
problem that runs
in time polynomial in $m$, $n$, and $D_{max}$, and achieves an approximation
factor that is within $O(\log (mD_{max}))$ of optimal,
where $D_{max}$ is the maximum value of  
$\sum_{i=1}^{n} |a_{j_i}|$ over all the functions $f_j$.
\end{theorem}
\begin{proof}
For each pair of linear equations $f_i$ and $f_j$ in the system, 
where $i < j$,
let $f_{ij}$ denote the linear function $f_i - f_j$.
We construct a utility function $g^{(ij)}$ with goal value $Q^{(ij)}$.
Intuitively, the
goal value of $g^{(ij)}$ is reached when there is enough
information to determine that $f_{ij}(x) \geq 0$, 
or when there is enough information to determine that
$f_{ij}(x) \leq 0$.

The construction of $g^{(ij)}$ is very similar to the construction
of the utility function in our first threshold evaluation algorithm.
For each $i,j$ pair, 
let $min_{ij}(b)$ be the minimum value of $f_{ij}(b')$
on any assignment $b' \in \{0,1\}^n$ such that $b' \sim b$,
and let $max_{ij}(b)$ be the maximum value.
Let $R_{max(ij)} = max_{ij}(*, \ldots, *)$
and
let $R_{min(ij)} = min_{ij}(*, \ldots, *)$.

Let
$g^{(ij)}_{<}:\{0,1,*\}^n \rightarrow \mathbb{Z}_{\geq 0}$, 
be defined as follows.
If $R_{max(ij)} \leq 0$, then
$g^{(ij)}_{<}(b) = 0$ for all $b \in \{0,1,*\}^n$
and $Q^{(ij)}_{<} = 0$.
Otherwise,
for $b \in \{0,1,*\}^n$,
let $g^{(ij)}_{<}(b)$ $= min\{R_{max(ij)},R_{max(ij)} - max_{ij}(b)\}$
and $Q^{(ij)}_{<} = R_{max(ij)}$.
It follows that for $b \in \{0,1,*\}^n$,
$f_i(b') \leq f_j(b')$ for all extensions $b' \sim b$
iff $g^{(ij)}_{<}(b) = Q^{(ij)}_{<}$.

We define
$g^{(ij)}_{>}$ 
and $Q^{(ij)}_{<}$ symmetrically,
so that $f_i(b') \geq f_j(b')$ for all extensions $b' \sim b$
iff $g^{(ij)}_{>}(b) = Q^{(ij)}_{>}$.

We apply the disjunctive construction of Lemma~\ref{combinegoals}
to combine $g^{(ij)}_{>}$ 
and
$g^{(ij)}_{<}$ and their associated goal values.  Let the resulting
new utility function be
$g^{(ij)}$ and let its goal 
value be $Q^{(ij)}$.
As in the analysis of the algorithm in Section~\ref{sec:thresholdQ},
we can show that $Q^{(ij)}$ is $O(D^2)$,
where $D$ is the sum of the magnitudes of the coefficients 
in $f_{ij}$.

Using the AND construction of Lemma~\ref{combinegoals}
to combine the $g^{(ij)}$
we get our final utility function $g = \sum_{i < j} g^{(ij)}$
with goal value $Q = \sum_{i < j} Q^{(ij)}$.

We now show that achieving the goal utility $Q$ is
equivalent to having enough information to do the ranking.
Until the goal value is reached, there is still a pair
$i,j$ such that it remains possible that
$f_i(x) > f_j(x)$ (under one setting of the untested
variables), and it remains possible that 
$f_j(x) < f_i(x)$ (under another setting).  In this situation, we do not have
enough information to output a ranking we know to be valid.

Once $g(b) = Q$, the situation changes.
For each $i,j$ such that $f_i(x) < f_j(x)$, 
we know that $f_i(x) \leq f_j(x)$. Similarly, if 
$f_i(x) > f_j(x)$, then at goal utility $Q$, we know that $f_i(x) \geq f_j(x)$.
If $f_i(x) = f_j(x)$ at goal utility $Q$, 
we may only know that $f_i(x) \geq f_j(x)$
or that $f_i(x) \leq f_j(x)$. 
We build a valid ranking from this knowledge as follows.
If there exists an $i$ such that we know that $f_i(x) \leq f_j(x)$
for all $j \neq i$, then we place $f_i(x)$ first in our ranking,
and recursively rank the other elements.
Otherwise, we can easily find a ``directed cycle,'' i.e. a sequence
$i_1, \ldots, i_m$, $m \geq 2$, such that we know that
$f_{i_1}(x) \leq f_{i_2}(x) \leq \ldots \leq f_{i_m}(x)$ and
$f_{i_m}(x) \leq f_{i_1}(x)$. It follows that
$f_{i_1}(x) = \ldots = f_{i_m}(x)$.
In this case, we can delete $f_{i_2}, \ldots, f_{i_m}$, 
recursively rank $f_{i_1}$ and the remaining $f_i$, and then insert
$f_{i_2}, \ldots, f_m$ into the ranking next to $f_{i_1}$.

Applying Adaptive Greedy to solve the SSSC problem for $g$, 
the theorem follows from the
$(\ln Q + 1)$ approximation bound for Adaptive Greedy,
and the fact that $Q = O(D_{max}^2m^2)$.
\end{proof}

\section{Acknowledgments}
Lisa Hellerstein was partially supported by NSF Grants 1217968 
and 0917153.
Devorah Kletenik was partially supported by 
NSF Grant 0917153 and by US Department of
Education GAANN Grant P200A090157.
Amol Deshpande was partially supported by NSF Grants 
0916736 and 1218367.
We would like to thank Tongu\c{c} {\"U}nl{\"u}yurt
and Sarah Allen for their helpful feedback on earlier versions this paper,
and Sarah Allen for preparing a notation summary.
We thank an anonymous referee for suggesting a way
to simplify the original analysis of ADG. Lisa Hellerstein would like to thank Endre Boros, Kazuhisa Makino,
and Vladimir Gurvich for a stimulating discussion at RUTCOR.

\bibliographystyle{custom}
\bibliography{soda,throughput}

\appendix
\section{Table of notation}
\label{append:tablenotation}
\begin{tabular} {l  l}
$x_i$ & the $i$th variable \\
$p_i$ & probability that variable $x_i$ is 1 \\
$c_i$ & cost of testing $x_i$ \\
$p$ & the probability product vector $(p_1, p_2, \ldots, p_n)$ \\
$c$ & the cost vector $(c_1, c_2, \ldots, c_n)$ \\
$b$ & a partial assignment, an element of $\{0,1,*\}^n$ \\
$dom(b)$ & $\{b_i | b_i \neq \ast\}$, the set of variables of $b$ that have already been tested \\
$a \sim b$ & $a$ extends $b$ (is identical to $b$ for all variables $i$ such that $b_i \neq *$)\\
$D_p$ & product distribution, defined by $p$ \\
$x \sim D_p$ & a random $x$ drawn from distribution $D_p$ \\
$Q$ &  goal utility\\
$P$  & maximum utility that testing a single variable $x_i$ can contribute \\
$g$ & utility function defined on partial assignments with a value in $\{0,\ldots, Q\}$ \\
$N$ & the set $\{1,\ldots,n\}$ \\
$S$ & a subset of $N$ \\
$g(S,b)$ &utility of testing only the items in $S$, with outcomes specified by $b$ \\
$g_{S,b}(j)$ & $g(S \cup \{j\}, b) - g(S, b)$ \\
$b_{x_i \gets l}$ & $b$ extended by testing variable $i$ with outcome $l$ \\
$k$ & number of clauses in a CNF \\
$d$ & number of terms in a DNF \\
$m$ & the number of linear threshold formulas in the simultaneous evaluation problem \\
$min(b)$ & the minimum possible value of the linear threshold function for any extension of $b$ \\
$max(b)$ & symmetric to $min(b)$, but maximum \\
$R_{min}$ & $min(*, \ldots, *)$ \\
$R_{max}$ & $max(*, \ldots, *)$ \\
$W$ & the set of partial assignments that contain exactly one $*$\\
$w^{(0)}, w^{(1)}$ & for $w \in W$, the extensions obtained from $w$ by setting the $\ast$ to 0 and 1, respectively\\
$j(w)$ & for $w \in W$, the $j$ for which $w_j = *$\\
$a^j$ & the partial assignment produced from $a$ by setting the $j$th bit to $*$ for assignment $a$\\
$a'$ & the assignment produced from $a$ by complementing the $j$th bit \\
$g_a(S)$ & $g(S,a)$\\
$y_{S,a}$ & the variable in LP2 for SSC associated with subset $S$ and assignment $a$ \\
$C(a)$ & the sequence of items tested by ADG on assignment $a$, in order of testing \\
$Y_{S,a}$ & the value of ADG variable $y_{S}$ after running ADG on input $a$  \\
$h'_w(y)$ & the left hand side of the constraint in LP2 for $w$ (a function of the $y_{S,a}$ variables) \\
$Y^t$ & assignment to the $y_{S,a}$ variables s.t. $y_{S,a}$ is the
value of ADG variable $y_S$ at the end of iteration $t$\\
& of its while loop, when ADG is run on input $a$ \\
$T^x$ & the number of iterations of the Adaptive Greedy (AG) while loop on input $x$ \\
$b^t_x$ & the value of $b$ on input $x$ after the $t$th iteration of the loop of AG on $x$ \\
$Y^x$ & the assignment to the  LP2 variables used in the analysis of the new bound for AG \\
$q^x(y)$ & $\sum_{S \subset N} g_{S,x}(N-S) y_S^x$ \\
$F^t$ &  variable of ADG, the set containing the first $t$ variables it tests \\
$g(j)$ & equals $\max_{l \in \{0,1\}} g_r(j,l)$ where $r = (*, \ldots, *)$, in analysis of Adaptive Greedy

%
%
%

\end{tabular}

\end{document}